\documentclass[11pt]{article}

\usepackage{fullpage}
\usepackage{amsmath,amsthm,amsfonts}
\usepackage{amssymb,latexsym,graphicx}
\usepackage{fourier}
\usepackage{mathtools}
\usepackage{hyperref}
\usepackage{authblk}
\usepackage{graphicx}
\usepackage{url}
\pdfoutput=1

\makeatletter
\def\resetMathstrut@{%
  \setbox\z@\hbox{%
    \mathchardef\@tempa\mathcode`\(\relax
    \def\@tempb##1"##2##3{\the\textfont"##3\char"}%
    \expandafter\@tempb\meaning\@tempa \relax
  }%
  \ht\Mathstrutbox@1.2\ht\z@ \dp\Mathstrutbox@1.2\dp\z@
}
\makeatother

\newtheorem{theorem}{Theorem}[section]

\newtheorem{lemma}[theorem]{Lemma}

\newtheorem{definition}[theorem]{Definition}

\newcommand{\R}{\ensuremath{\mathbb{R}}}

\renewcommand{\subset}{\subseteq}

\newcommand{\eps}{\varepsilon}
\renewcommand{\epsilon}{\varepsilon}

\renewcommand{\leq}{\leqslant}
\renewcommand{\geq}{\geqslant}

\newcommand{\E}{\mathbb{E}}

  \DeclareMathOperator{\tr}{tr}

\begin{document}

\title{Satisfying the Restricted Isometry Property with the Optimal Number of Rows and Slightly Less Randomness}
\author[1]{Shravas Rao \footnote{Funding: This material is based upon work supported by the National Science Foundation under Award No. 2348489}}
\affil[1]{Department of Computer Science, Portland State University \protect \\
1900 SW 4th Ave, Portland, OR 97201 \protect \\
shravas@pdx.edu}
\maketitle

\begin{abstract}
A matrix $\Phi \in \R^{Q \times N}$ satisfies the restricted isometry property if $\|\Phi x\|_2^2$ is approximately equal to $\|x\|_2^2$ for all $k$-sparse vectors $x$.
We give a construction of RIP matrices with the optimal $Q = O(k \log(N/k))$ rows using $O(k\log(N/k)\log(k))$ bits of randomness.
The main technical ingredient is an extension of the Hanson-Wright inequality to $\eps$-biased distributions.
\end{abstract}


\section{Introduction}

A matrix $\Phi \in \R^{Q \times N}$ is said to satisfy the $(k, \eta)$-restricted isometry property if for every $k$-sparse vector $x \in \R^{N}$, one has
\[
(1-\eta)\|x\|_2^2 \leq \|\Phi x\|_2^2 \leq (1+\eta)\|x\|_2^2.
\]
where $\|\cdot\|_2$ denotes the $\ell_2$ norm.
This notion, introduced by Cand\`es and Tao~\cite{CT05}, has many applications especially in the field of compressed sensing~\cite{C08}.

A major goal is to construct such matrices so that the number of rows $Q$ is as small as possible compared to $k$ and $N$.
In this note, we will mostly ignore the dependence on $\eta$.
It is known that $Q = \Omega(k \log(N/k))$ rows are necessary for the restricted isometry property to hold~\cite{FPRU10}.

If the construction is required to be deterministic there are many constructions with $Q = \widetilde{O}(k^2)$ rows~\cite{K75, AGHP92, D07, NT11}.
In a breakthrough result due to Bourgain, Dilworth, Ford, Konyagin, and Kutzarova~\cite{BDFKK11}, a construction with $O(k^{2-\eps_0})$ rows was obtained, for some very small constant $\eps_0$.
This was later improved by Mixon~\cite{M15}, and Bandeira, Mixon, and Moreira~\cite{BMM17} conditioned on a number-theoretic conjecture.
It is known that significant improvements would imply explicit constructions of Ramsey graphs~\cite{G20}.

On the other hand, there exist randomized constructions that do achieve the optimal number of rows.
In particular, if each entry is an independent Gaussian or Bernoulli random variable, then the restricted isometry property holds for $Q = O(k \log(N/k))$~\cite{CT06, BDDW08, MPTJ08}.

Randomized constructions can be evaluated by the number of bits of randomness they use.
When all entries are independent, then $O(NQ)$ bits of randomness are needed.
Thus, towards the goal of deterministic constructions of matrices that satisfy the restricted isometry property, one can ask if there exist constructions that use fewer bits of randomness.

There is a long line of work that considers a construction where each row of $\Phi$ is a random row of a Fourier or Walsh-Hadamard matrix~\cite{CT06, RV08, CGV13, B14, HR17}.
Such constructions use $O(Q\log(N))$ bits of randomness.
The analysis due to Haviv and Regev~\cite{HR17}, and later improved by~\cite{BDJR21}, obtain the best-known bound on the number of rows $Q = O(k\log(N)\log^2(k))$.
However, it is known that when the rows are obtained from a Walsh-Hadamard matrix, the number of rows must be $Q = \Omega(k \log(N) \log(k))$~\cite{BLLMR19}.
Thus, such constructions can not achieve the optimal number of rows that constructions with independent entries can.

Another choice of construction allows for the entries of $\Phi$ to come from a $O(Q)$-wise independent distribution.
Combining the analysis of~\cite{CW09} with~\cite{BDDW08}, a slight modification to the analysis for independent entries holds for $O(Q)$-wise independent entries.
In particular, one can let $Q$ be the optimal $O(k\log(N/k))$.
Standard constructions of $O(Q)$-wise independent random variables when $Q = O(k\log(N/k))$ require $O(k\log(N/k) \log(N))$ bits of randomness.
This was slightly improved in~\cite{KN10} to a construction that requires only $O(k\log(N/k) \log(k\log(N/k)))$ bits of randomness.

Finally, a third line of work uses pseudorandom properties of the Legendre symbol.
This approach uses $O(k\log(N)\log(k))$ bits of randomness, but requires $Q = O(k\log(N)\log^2(k))$ rows~\cite{BFFM16}.

In this note, we use almost $O(Q)$-wise independent distributions when $Q = O(k \log(N/k))$ to show that that $O(k \log(N/k)\log(k))$ random bits is enough to construct a matrix that satisfies the restricted isometry property.
This is the least amount of randomness used among all randomized constructions with the optimal number of rows listed above.
However, when $k \sim N^c$ for some constant $c$ for example, this recovers the result in~\cite{KN10} as $\log(k) = \Theta(\log(N))$ in this case.

The pseudorandom properties of the almost $O(Q)$-wise independent distributions that we use are similar to the pseudorandom properties of the Legendre symbol used in~\cite{BFFM16}.
Compared to~\cite{KN10} and~\cite{BFFM16}, our proof is arguably simpler.
However, it is conjectured that the Legendre symbol can be used to construct deterministic RIP matrices, something that is not true for the techniques in this note.

\begin{theorem}\label{thm:main}
There exists a distribution of $Q \times N$ matrices for $Q = O(k\log(N/k)\eta^{-2})$ that can be sampled efficiently using $O(k\log(k)\log(N/k))$ bits of randomness such that a sample $\Phi$ from this distribution satisfies the $(k, \eta)$-restricted isometry property with high probability.
\end{theorem}

The proof of this theorem follows the same structure as~\cite{CW09} with~\cite{BDDW08} when the entries of $\Phi$ come from a truly $O(Q)$-wise independent distribution.
In particular, one focuses on the more general question of constructing a matrix $\Phi$ that is a Johnson-Lindenstrauss projection.
That is, one shows that for any fixed unit vector $x$
\begin{equation}\label{eq:jl}
\Pr\left[\left|\|\Phi x\|_2^2 -1\right| \geq \eta \right] \leq \exp(-cQ\eta^2)
\end{equation}
for some constant $c$.
It was shown in~\cite{BDDW08} that if Eq.~\eqref{eq:jl} holds for $Q = O(k\log(N/k))$ when $\eta < 1$, then the random matrix $\Phi$ satisfies the restricted isometry property with high probability.

In this note, we only consider the case of when $x$ is $k$-sparse.
That is, when $x$ is not $k$-sparse, and the entries of $\Phi$ come from an almost $O(Q)$-wise independent distribution, the number of random bits required for our generalization of the Hanson-Wright inequality to hold is too large to improve upon the result in~\cite{CW09}.
However, for the purposes of constructing a matrix that satisfies the restricted isometry property, it is enough to consider only $k$-sparse $x$.
This is because for the restricted isometry property to hold, it is enough to show that $\|\phi x\|_2^2 \approx \|x\|_2^2$ for $k$-sparse vectors $x$.

\section{Preliminaries}

If $x \in \R^N$ is a vector, we define the norm
\[
\|x\|_p^p = \sum_{i=1}^N |x(i)|^p.
\]
If $A \in \R^{N \times N}$ is a matrix, we define the norms
\[
\|A\|_{L_1} = \sum_{i, j} |A_{i, j}|, \;\; \|A\|_2 = \max_{\|x\|_2 = 1} \|Ax\|_2, \;\; \|A\|_{S_2} = (\tr(A A^*))^{1/2}.
\]

The construction of matrices is derived from $\eps$-biased distributions, which we define below.

\begin{definition}
A distribution $S \subset \{-1, 1\}^n$ is $\eps$-biased $\ell$-wise independent if when $x = (x_1, x_2, \ldots, x_n)$ is sampled uniformly from $S$, for all $y \subset [n]$ such that $0 < |y| \leq \ell$,
\[
\left|\E\left[\prod_{i \in y} x_i\right]\right| \leq \eps.
\]
\end{definition}

There exist $\eps$-biased $\ell$-wise independent distributions $S$ such that $\log(|S|) = O(\log(\log(n))+\ell+\log(1/\eps))$, due to~\cite{AGHP92, NN93}.

The main ingredient of the proof of Theorem~\ref{thm:main} is a generalization of the Hanson-Wright inequality.
This inequality can often be used to obtain concentration inequalities for quadratic forms using Markov's inequality, as will be done here.
We state the original below~\cite{HW71} (see Theorem 3.1 in~\cite{KN12}).

\begin{theorem}\label{thm:hw}
There exists a constant $C$ such that the following holds.
Let $(z_1, \ldots, z_n)$ be independent random variables uniform over $\{-1, 1\}$.
Then for any symmetric $B \in \R^{n \times n}$ and integer $\ell \geq 2$ a power of $2$,
\[
\E\left[\left(z^T B z - \tr(B)\right)^{\ell}\right] \leq 
C^{\ell} \max\left\{\sqrt{\ell} \cdot \|B\|_{S_2}, \ell \cdot \|B\|_2 \right\}^{\ell}\]
\end{theorem}


\section{Proof of the Main Theorem}

The main technical ingredient of this note is a generalization of the Hanson-Wright inequality~\cite{HW71} for $\eps$-biased distributions which we state and prove below.

\begin{lemma}\label{lem:hwbiased}
There exists a constant $C$ such that the following holds.
Let $(z_1, \ldots, z_n)$ be a sample from an $\eps$-biased $2\ell$-wise independent distribution from $\{1, -1\}^n$ for any integer $\ell \geq 2$ a power of $2$.
Then for any symmetric $B \in \R^{n \times n}$,
\[
\E\left[\left(z^T B z - \tr(B)\right)^{\ell}\right] \leq 
\eps \|B\|_{L_1}^{\ell}+
C^{\ell} \max\left\{\sqrt{\ell} \cdot \|B\|_{S_2}, \ell \cdot \|B\|_2 \right\}^{\ell}\]
\end{lemma}
\begin{proof}
Let $\mathcal{S} = \{(i, j) \mid i, j \in [n], i \neq j\}$.
Then,
\begin{equation}\label{eq:express}
(z^T B z - \tr(B))^{\ell} = \sum_{s \in \mathcal{S}^{\ell}} \prod_{k = 1}^{\ell} z_{s_{k, 1}}z_{s_{k, 2}}B_{s_{k, 1}, s_{k, 2}}.
\end{equation}
Let $\mathcal{S}'$ be the set of sequences in $\mathcal{S}^{\ell}$ such that each $i \in [n]$ appears in an even number of pairs.
Then Eq.~\eqref{eq:express} can be rewritten as
\begin{align}
 \left(\sum_{s \in \mathcal{S}^{\ell} \backslash \mathcal{S}'} \prod_{k = 1}^{\ell} z_{s_{k, 1}}z_{s_{k, 2}}B_{s_{k, 1}, s_{k, 2}}\right)&+
 \left(\sum_{s \in \mathcal{S}'} \prod_{k = 1}^{\ell} z_{s_{k, 1}}z_{s_{k, 2}}B_{s_{k, 1}, s_{k, 2}}\right) \nonumber \\
&= 
 \left(\sum_{s \in \mathcal{S}^{\ell} \backslash \mathcal{S}'} \prod_{k = 1}^{\ell} z_{s_{k, 1}}z_{s_{k, 2}}B_{s_{k, 1}, s_{k, 2}}\right)+
 \left(\sum_{s \in \mathcal{S}'} \prod_{k = 1}^{\ell} B_{s_{k, 1}, s_{k, 2}}\right) \label{eq:separate} \\
&\leq
 \left(\sum_{s \in \mathcal{S}^{\ell} \backslash \mathcal{S}'} \prod_{k = 1}^{\ell} z_{s_{k, 1}}z_{s_{k, 2}}B_{s_{k, 1}, s_{k, 2}}\right)
+
C^{\ell} \max\left\{\sqrt{\ell} \cdot \|B\|_{S_2}, \ell \cdot \|B\|_2 \right\}^{\ell}. \nonumber
\end{align}
The equality follows from the definition of $\mathcal{S}$ and the fact that $z_i^2 = 1$ for all $i$.
The inequality follows from the fact that when the $z_i$ are independent, the first sum in Eq.~\eqref{eq:separate} evaluates to $0$, and thus, Theorem~\ref{thm:hw} can be used to bound the second sum.

Because the $z_i$ come from an $\eps$-biased $2\ell$-wise independent distribution, 
\[
\left|\prod_{k = 1}^{\ell} z_{s_{k, 1}}z_{s_{k, 2}}\right| \leq \eps
\]
when $s \not\in \mathcal{S}'$.
Thus, the left-hand side of Eq.~\eqref{eq:separate} is bounded above by
\begin{align*}
\eps \left(\sum_{s \in \mathcal{S}^{\ell} \backslash \mathcal{S}'} \prod_{k = 1}^{\ell} \left|B_{s_{k, 1}, s_{k, 2}}\right|\right)
&+
C^{\ell} \max\left\{\sqrt{\ell} \cdot \|B\|_{S_2}, \ell \cdot \|B\|_2 \right\}^{\ell} \\
&\leq
\eps \left(\sum_{s \in \mathcal{S}^{\ell}} \prod_{k = 1}^{\ell} \left|B_{s_{k, 1}, s_{k, 2}}\right|\right)
+
C^{\ell} \max\left\{\sqrt{\ell} \cdot \|B\|_{S_2}, \ell \cdot \|B\|_2 \right\}^{\ell}
\\
&=
\eps \|B\|_{L_1}^{\ell}+
C^{\ell} \max\left\{\sqrt{\ell} \cdot \|B\|_{S_2}, \ell \cdot \|B\|_2 \right\}^{\ell}
\end{align*}
as desired.
\end{proof}

We now prove Theorem~\ref{thm:main} using the same main ideas as in ~\cite[Theorem 6]{KN10}.

\begin{proof}[Proof of Theorem~\ref{thm:main}]
We let $\Phi$ come from a $(k^2\ell/Q)^{-\ell/2}$-biased $\ell$-wise independent distribution for $\ell = O(Q \eta^2) = O(k \log(N/k))$, normalized so that all entries are from the set $\{-1/Q, 1/Q\}$.
There exists a construction using the desired number of random bits, $O(\ell \log(k)) = O(k\log(k)\log(N/k))$, due to~\cite{AGHP92, NN93}.

By Theorem 5.2 in~\cite{BDDW08}, it is enough to show that Eq.~\eqref{eq:jl} holds when $x$ is a $k$-sparse vector such that $\|x\|_2 =1$, and $\eta < 1$.

For every $k$-sparse vector $x$, let $B_x$ be a $Q N \times Q N$ block-diagonal matrix with $Q$ blocks, where each block is equal to $x x^T / Q$.
Let $z \in \R^{QN}$ contain each row of $\Phi$ stacked in a vector.
Note that $B_x$ contains at most $k^2 Q$ non-zero entries as $x$ is $k$-sparse, and thus $\|B_x\|_{L_1} \leq k Q^{1/2} \|B_x\|_{S_2}$.
By Lemma~\ref{lem:hwbiased}, for some constants $C_1$ and $C_2$,
\begin{align*}
\E\left[\left(z^T B_x z - \tr(B_x)\right)^{\ell}\right]  &\leq \ell^{\ell/2} \|B_x\|_{S_2}^{\ell} +C_1^{\ell} \max\left\{\sqrt{\ell} \cdot \|B_x\|_{S_2}, \ell \cdot \|B_x\|_2 \right\}^{\ell} \\
&\leq C_2^{\ell} \max\left\{\sqrt{\ell} \cdot \|B_x\|_{S_2}, \ell \cdot \|B_x\|_2 \right\}^{\ell}
\end{align*}
By direct computation, $z^T B_x z = \|\Phi x\|_2^2$, $\tr(B_x) = 1$, $\|B_x\|_{S_2}^2 = Q \|x/Q\|_2^4 = 1/Q$ and $\|B_x\|_2 = \| x^T x/Q\|_2 = 1/Q$.
Thus, the bound becomes
\[
\E\left[\left(\left\|\Phi x\right\|_2^2-1\right)^{\ell}\right] \leq C_2^{\ell} \max\left\{\sqrt{\ell/Q}, \ell/Q\right\}^{\ell}.
\]
By Markov's inequality, one has
\begin{align*}
\Pr\left[\left(\left\|\Phi x\right\|_2^2-1\right)^2 \geq \eta^2\right] 
&= \Pr\left[\left(\left\|\Phi x\right\|_2^2-1\right)^{\ell} \geq \eta^{\ell}\right] \\
&\leq C_2^{\ell} \eta^{-\ell} \max\left\{\sqrt{\ell/Q}, (\ell/Q)\right\}^{\ell} \\
&\leq \exp(- c Q \eta^2)
\end{align*}
for some constant $c$ and $\ell = O(Q \eta^2)$, where the last inequality follows by noting that $\eta < 1$.
\end{proof}

\bibliographystyle{alphaabbrv}
\bibliography{riplb}

\newcommand{\etalchar}[1]{$^{#1}$}
\begin{thebibliography}{BDDW08}
\expandafter\ifx\csname urlstyle\endcsname\relax
  \providecommand{\doi}[1]{doi:\discretionary{}{}{}#1}\else
  \providecommand{\doi}{doi:\discretionary{}{}{}\begingroup
  \urlstyle{rm}\Url}\fi

\bibitem[AGHP92]{AGHP92}
N.~Alon, O.~Goldreich, J.~Håstad, and R.~Peralta.
\newblock Simple constructions of almost {$k$}-wise independent random
  variables.
\newblock \emph{Random Structures Algorithms}, 3(3):289--304, 1992.
\newblock ISSN 1042-9832,1098-2418.
\newblock \doi{10.1002/rsa.3240030308}.

\bibitem[BDDW08]{BDDW08}
R.~Baraniuk, M.~Davenport, R.~DeVore, and M.~Wakin.
\newblock A simple proof of the restricted isometry property for random
  matrices.
\newblock \emph{Constr. Approx.}, 28(3):253--263, 2008.
\newblock ISSN 0176-4276,1432-0940.
\newblock \doi{10.1007/s00365-007-9003-x}.

\bibitem[BDF{\etalchar{+}}11]{BDFKK11}
J.~Bourgain, S.~Dilworth, K.~Ford, S.~Konyagin, and D.~Kutzarova.
\newblock Explicit constructions of {RIP} matrices and related problems.
\newblock \emph{Duke Math. J.}, 159(1):145--185, 2011.
\newblock ISSN 0012-7094,1547-7398.
\newblock \doi{10.1215/00127094-1384809}.

\bibitem[BDJR21]{BDJR21}
S.~Brugiapaglia, S.~Dirksen, H.~C. Jung, and H.~Rauhut.
\newblock Sparse recovery in bounded {R}iesz systems with applications to
  numerical methods for {PDE}s.
\newblock \emph{Appl. Comput. Harmon. Anal.}, 53:231--269, 2021.
\newblock ISSN 1063-5203,1096-603X.
\newblock \doi{10.1016/j.acha.2021.01.004}.

\bibitem[BFMM16]{BFFM16}
A.~S. Bandeira, M.~Fickus, D.~G. Mixon, and J.~Moreira.
\newblock Derandomizing restricted isometries via the {L}egendre symbol.
\newblock \emph{Constr. Approx.}, 43(3):409--424, 2016.
\newblock ISSN 0176-4276,1432-0940.
\newblock \doi{10.1007/s00365-015-9310-6}.

\bibitem[BLL{\etalchar{+}}23]{BLLMR19}
J.~Blasiok, P.~Lopatto, K.~Luh, J.~Marcinek, and S.~Rao.
\newblock An improved lower bound for sparse reconstruction from subsampled
  {W}alsh matrices.
\newblock \emph{Discrete Anal.}, pages Paper No. 3, 9, 2023.
\newblock ISSN 2397-3129.

\bibitem[BMM17]{BMM17}
A.~S. Bandeira, D.~G. Mixon, and J.~Moreira.
\newblock A conditional construction of restricted isometries.
\newblock \emph{Int. Math. Res. Not. IMRN}, (2):372--381, 2017.
\newblock ISSN 1073-7928,1687-0247.
\newblock \doi{10.1093/imrn/rnv385}.

\bibitem[Bou14]{B14}
J.~Bourgain.
\newblock An improved estimate in the restricted isometry problem.
\newblock In \emph{Geometric aspects of functional analysis}, volume 2116 of
  \emph{Lecture Notes in Math.}, pages 65--70. Springer, Cham, 2014.
\newblock \doi{10.1007/978-3-319-09477-9_5}.

\bibitem[Can08]{C08}
E.~J. Cand\`es.
\newblock The restricted isometry property and its implications for compressed
  sensing.
\newblock \emph{C. R. Math. Acad. Sci. Paris}, 346(9-10):589--592, 2008.
\newblock ISSN 1631-073X.
\newblock \doi{10.1016/j.crma.2008.03.014}.

\bibitem[CGV13]{CGV13}
M.~Cheraghchi, V.~Guruswami, and A.~Velingker.
\newblock Restricted isometry of {F}ourier matrices and list decodability of
  random linear codes.
\newblock \emph{SIAM J. Comput.}, 42(5):1888--1914, 2013.
\newblock ISSN 0097-5397.
\newblock \doi{10.1137/120896773}.

\bibitem[CT05]{CT05}
E.~J. Candes and T.~Tao.
\newblock Decoding by linear programming.
\newblock \emph{IEEE Trans. Inform. Theory}, 51(12):4203--4215, 2005.
\newblock ISSN 0018-9448.
\newblock \doi{10.1109/TIT.2005.858979}.

\bibitem[CT06]{CT06}
E.~J. Candes and T.~Tao.
\newblock Near-optimal signal recovery from random projections: universal
  encoding strategies?
\newblock \emph{IEEE Trans. Inform. Theory}, 52(12):5406--5425, 2006.
\newblock ISSN 0018-9448.
\newblock \doi{10.1109/TIT.2006.885507}.

\bibitem[CW09]{CW09}
K.~L. Clarkson and D.~P. Woodruff.
\newblock Numerical linear algebra in the streaming model.
\newblock In M.~Mitzenmacher, editor, \emph{Proceedings of the 41st Annual
  {ACM} Symposium on Theory of Computing, {STOC} 2009, Bethesda, MD, USA, May
  31 - June 2, 2009}, pages 205--214. {ACM}, 2009.
\newblock \doi{10.1145/1536414.1536445}.

\bibitem[DeV07]{D07}
R.~A. DeVore.
\newblock Deterministic constructions of compressed sensing matrices.
\newblock \emph{J. Complexity}, 23(4-6):918--925, 2007.
\newblock ISSN 0885-064X,1090-2708.
\newblock \doi{10.1016/j.jco.2007.04.002}.

\bibitem[FPRU10]{FPRU10}
S.~Foucart, A.~Pajor, H.~Rauhut, and T.~Ullrich.
\newblock The {G}elfand widths of {$\ell_p$}-balls for {$0<p\leq 1$}.
\newblock \emph{J. Complexity}, 26(6):629--640, 2010.
\newblock ISSN 0885-064X.
\newblock \doi{10.1016/j.jco.2010.04.004}.

\bibitem[Gam20]{G20}
D.~Gamarnik.
\newblock Explicit construction of {RIP} matrices is {R}amsey-hard.
\newblock \emph{Comm. Pure Appl. Math.}, 73(9):2043--2048, 2020.
\newblock ISSN 0010-3640,1097-0312.
\newblock \doi{10.1109/tit.2014.2331341}.

\bibitem[HR17]{HR17}
I.~Haviv and O.~Regev.
\newblock The restricted isometry property of subsampled {F}ourier matrices.
\newblock In \emph{Geometric aspects of functional analysis}, volume 2169 of
  \emph{Lecture Notes in Math.}, pages 163--179. Springer, Cham, 2017.

\bibitem[HW71]{HW71}
D.~L. Hanson and F.~T. Wright.
\newblock A bound on tail probabilities for quadratic forms in independent
  random variables.
\newblock \emph{Ann. Math. Statist.}, 42:1079--1083, 1971.
\newblock ISSN 0003-4851.
\newblock \doi{10.1214/aoms/1177693335}.

\bibitem[Kas75]{K75}
B.~S. Kashin.
\newblock The diameters of octahedra.
\newblock \emph{Uspehi Mat. Nauk}, 30(4(184)):251--252, 1975.
\newblock ISSN 0042-1316.

\bibitem[KN10]{KN10}
D.~M. Kane and J.~Nelson.
\newblock A derandomized sparse {J}ohnson-{L}indenstrauss transform.
\newblock \emph{arXiv:1006.3585}, 2010.

\bibitem[KN14]{KN12}
D.~M. Kane and J.~Nelson.
\newblock Sparser {J}ohnson-{L}indenstrauss transforms.
\newblock \emph{J. ACM}, 61(1):Art. 4, 23, 2014.
\newblock ISSN 0004-5411,1557-735X.
\newblock \doi{10.1145/2559902}.

\bibitem[Mix15]{M15}
D.~G. Mixon.
\newblock Explicit matrices with the restricted isometry property: breaking the
  square-root bottleneck.
\newblock In \emph{Compressed sensing and its applications}, Appl. Numer.
  Harmon. Anal., pages 389--417. Birkh\"{a}user/Springer, Cham, 2015.
\newblock ISBN 978-3-319-16041-2; 978-3-319-16042-9.

\bibitem[MPTJ08]{MPTJ08}
S.~Mendelson, A.~Pajor, and N.~Tomczak-Jaegermann.
\newblock Uniform uncertainty principle for {B}ernoulli and subgaussian
  ensembles.
\newblock \emph{Constr. Approx.}, 28(3):277--289, 2008.
\newblock ISSN 0176-4276.
\newblock \doi{10.1007/s00365-007-9005-8}.

\bibitem[NN93]{NN93}
J.~Naor and M.~Naor.
\newblock Small-bias probability spaces: efficient constructions and
  applications.
\newblock \emph{SIAM J. Comput.}, 22(4):838--856, 1993.
\newblock ISSN 0097-5397.
\newblock \doi{10.1137/0222053}.

\bibitem[NT11]{NT11}
J.~L. Nelson and V.~N. Temlyakov.
\newblock On the size of incoherent systems.
\newblock \emph{J. Approx. Theory}, 163(9):1238--1245, 2011.
\newblock ISSN 0021-9045,1096-0430.
\newblock \doi{10.1016/j.jat.2011.04.001}.

\bibitem[RV08]{RV08}
M.~Rudelson and R.~Vershynin.
\newblock On sparse reconstruction from {F}ourier and {G}aussian measurements.
\newblock \emph{Comm. Pure Appl. Math.}, 61(8):1025--1045, 2008.
\newblock ISSN 0010-3640.
\newblock \doi{10.1002/cpa.20227}.

\end{thebibliography}

\end{document}